\documentclass[twocolumn]{article}
\usepackage[latin9]{inputenc}
\usepackage{geometry}
\usepackage{fancyhdr}
\usepackage{color}
\usepackage{amsthm}
\usepackage{amsmath}
\usepackage{amssymb}
\usepackage{graphicx}
\usepackage[unicode=true,pdfusetitle,
 bookmarks=true,bookmarksnumbered=false,bookmarksopen=false,
 breaklinks=false,pdfborder={0 0 1},backref=false,colorlinks=true]
 {hyperref}

\makeatletter
\theoremstyle{plain}
\newtheorem{thm}{\protect\theoremname}
  \theoremstyle{definition}
  \newtheorem{problem}[thm]{\protect\problemname}
  \theoremstyle{definition}
  \newtheorem{defn}[thm]{\protect\definitionname}
  \theoremstyle{remark}
  \newtheorem{claim}[thm]{\protect\claimname}

\makeatother

  \providecommand{\claimname}{Claim}
  \providecommand{\definitionname}{Definition}
  \providecommand{\problemname}{Problem}
\providecommand{\theoremname}{Theorem}

\begin{document}

\title{Reducing Linear Programs into Min-max Problems}

\author{Carmi Grushko}
\maketitle
\begin{abstract}
We show how to reduce a general, strictly-feasible LP problem, into
a min-max problem, which can be solved by the algorithm from the third
section of \cite{Grushko2012}.
\end{abstract}

\section{Reduction}
\begin{problem}
\label{LP-problem}Let us consider a linear program in the following
form,
\end{problem}
\begin{eqnarray*}
\mbox{maximize}_{x\in\mathbb{R}^{d}} &  & \left(0,0,\dots,1\right)^{T}x\\
\mbox{subject to} &  & Ax\leq b
\end{eqnarray*}

and let us assume that the problem is strictly feasible; that is,
there exists a point $x$ for which $Ax<b$. Further assume that the
origin $\left(0,0,\dots,0\right)$ is a strictly feasible point. 

Any strictly feasible linear program can be rotated such that the
objective function is $\left(0,0,\dots,1\right)\cdot x$, and translated
such that the origin is a strictly feasible point. The translation
is discussed in \emph{Strict Feasibility of the Origin}, below, while
the rotation is explained further below, in \emph{Rotation}.

We now show how to solve Problem \ref{LP-problem} using the algorithm
described in the third section of \cite{Grushko2012}.
\begin{defn}[$z$-axis]
\label{z-axis} Let $z$ denote the last coordinate of the space
of our problem. $z$-intersect of a hyperplane refers to its intersection
with the $z$-axis, while the last coordinate of a point is its $z$
value. For example, in a 5-dimensional space, the $z$-coordinate
denotes the fifth coordinate.
\end{defn}

\begin{defn}[Planes]
\label{def-planes}For terseness, we denote the plane $\pi\cdot p=\sigma$
as $\left(\pi,\sigma\right)$.
\end{defn}

\begin{defn}[Projective Duality]
\label{def-projective-duality}Let $p\in\mathbb{R}^{d}$ be a point.
Then its dual, $p^{*}$, is the plane $\left(p,-1\right)$. Conversely,
let $\left(\pi,\sigma\right)$ be a plane with $\sigma\neq0$. Then
its dual, $\left(\pi,\sigma\right)^{*}$, is the point $-\frac{\pi}{\ensuremath{\sigma}}$.
\end{defn}
It is straightforward to confirm that the projective duality is self-dual
and incidence preserving. For future use, we note that the $z-$intersect
of a dual $p^{*}$ to a point $p$ is $-\frac{1}{\ensuremath{p_{d}}}$.
\begin{defn}[Constraints and their Duals]
\label{def-constraints-and-duals}The set of constrains in Problem
\ref{LP-problem}, $Ax\leq b$, can be described by a set of planes.
Let us denote these planes as the set $\Pi={(A_{i},b_{i})}$, and
their duals as $\Pi^{*}=-{A_{i}/b_{i}}$.
\end{defn}
Note that we exclude the definition of duality for planes which intersect
the origin; however, since the origin is strictly feasible in Problem
\ref{LP-problem}, no constraint plane intersects it.
\begin{claim}
\label{claim-p-and-origin-same-side}Let $p$ be a point and $\left(\pi,\sigma\right)$
a plane. Then $p$ and the origin are on the same side of $\left(\pi,\sigma\right)$,
if and only if the point $\left(\pi,\sigma\right)^{*}$ and the origin
are on the same side of the plane $p^{*}$.
\begin{proof}
$p$ and the origin are on the same side of $\left(\pi,\sigma\right)$
iff,
\begin{eqnarray*}
\mbox{sign}\left(\left(\pi\cdot p-\sigma\right)\cdot\left(\pi\cdot0-\sigma\right)\right) & = & \mbox{sign}\left(\sigma^{2}\cdot\left(-\frac{1}{\sigma}\cdot\pi\cdot p+1\right)\right)\\
 & = & \mbox{sign}\left(-\frac{1}{\sigma}\cdot\pi\cdot p+1\right)\\
 & = & 1
\end{eqnarray*}
Similarly, $\left(\pi,\sigma\right)^{*}$ (which equals $-\pi/\sigma$)
and the origin are on the same side of the plane $p^{*}$ (which equals
$\left(p,-1\right)$) iff,
\[
\mbox{sign}\left(\left(p\cdot\left(-\frac{\pi}{\sigma}\right)+1\right)\cdot\left(p\cdot0+1\right)\right)=\mbox{sign}\left(-\frac{\pi}{\sigma}\cdot p+1\right)=1.
\]

\end{proof}
\end{claim}
This leads us to a characterization of the dual to a feasible point.
\begin{claim}
\label{claim-feasibility-of-dual}Assume the origin is a feasible
point. Then, a point $p$ is feasible iff the set of points $\Pi^{*}$
representing the problem constraints, and the origin, are on the same
side of the point's dual plane, $p^{*}$.\end{claim}
\begin{proof}
Since the origin is feasible, any other feasible point must share
with it the same side of all the constraint planes $\Pi$. By Claim
\ref{claim-p-and-origin-same-side}, this implies all duals to these
planes, $\Pi^{*}$, and the origin, must be on the same side of $p^{*}$.
\end{proof}
Claim \ref{claim-feasibility-of-dual} is illustrated in Figure \ref{fig:geometric-intuition}.
\begin{figure}
\begin{centering}
\includegraphics[width=0.45\columnwidth]{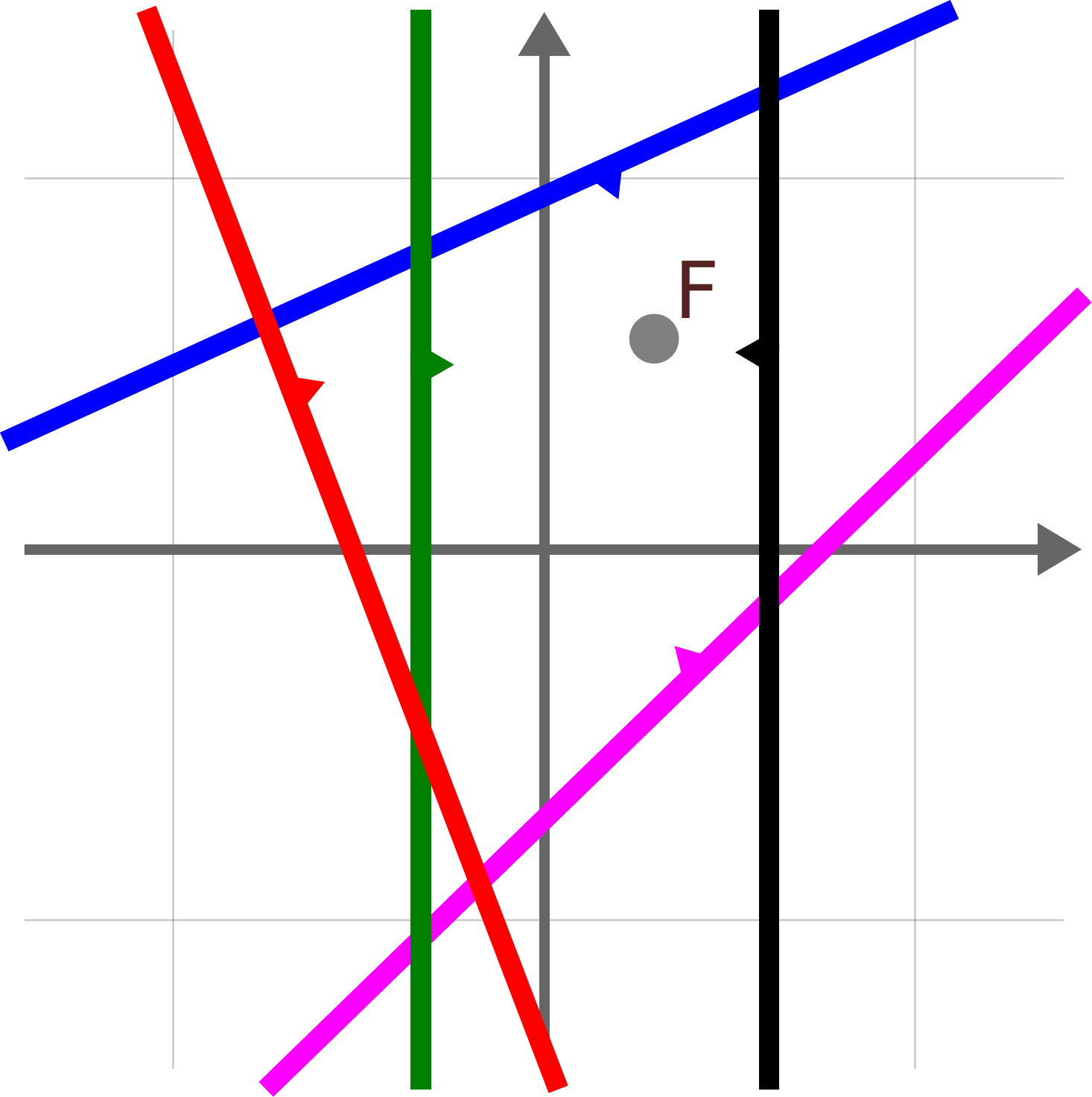}\hfill{}\includegraphics[width=0.45\columnwidth]{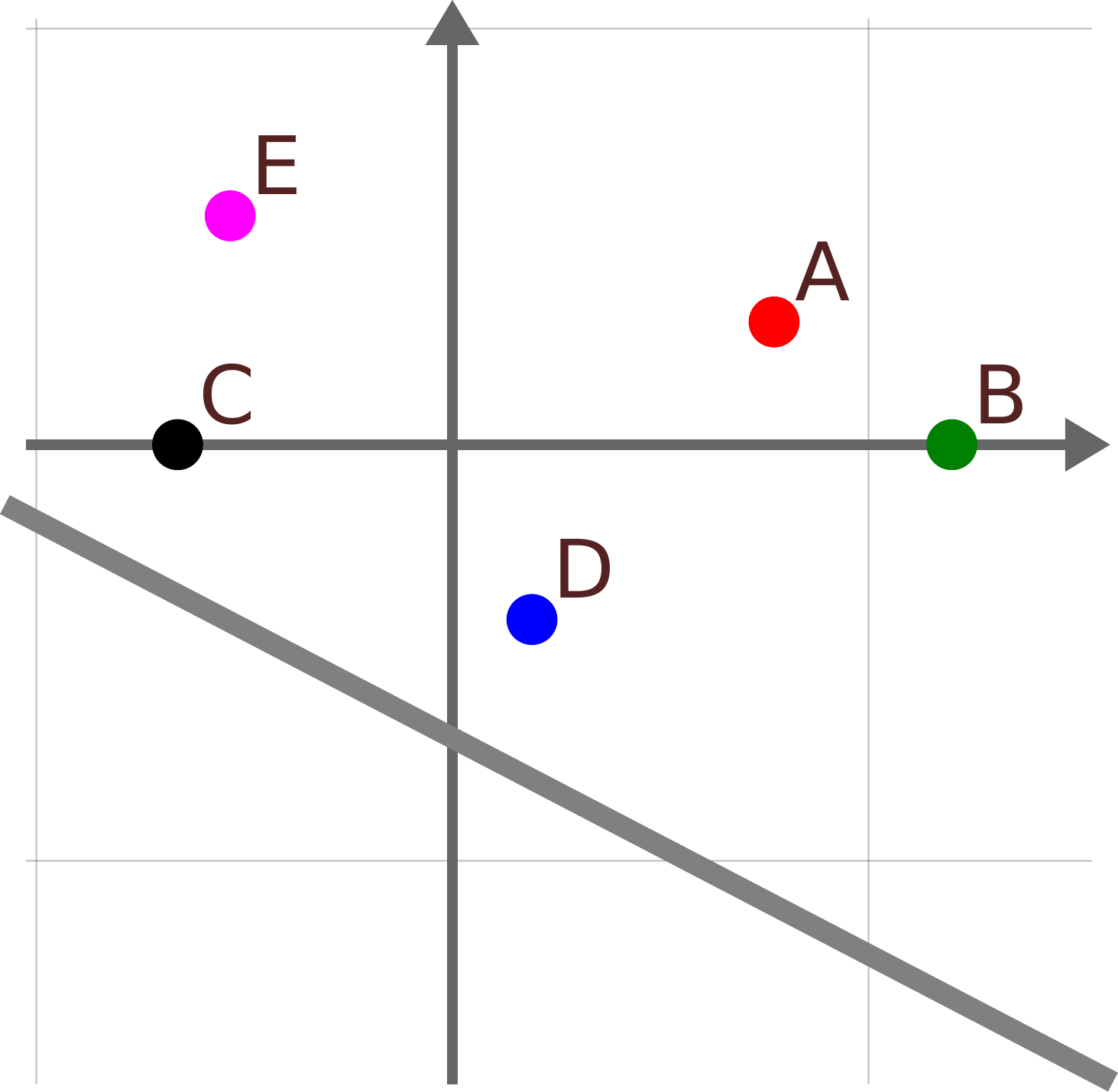}
\par\end{centering}

\caption{\emph{\label{fig:geometric-intuition}Left:} A point and a set of
lines. Little arrows denote feasible half-space.\emph{ Right:} The
duals to these lines and point. Elements and their duals are related
by color. }
\end{figure}
 The point $F$ is a feasible point and is on the same side as the
origin relative to all of the constraint planes (left figure). Its
dual, $F^{*}$ has all the constraint points $\Pi^{*}$ and the origin
on its same side (right figure).
\begin{defn}[Feasible Dual Plane]
A plane $\left(\pi,\sigma\right)$ is feasible if its dual point,
$\pi/\sigma$ is a feasible solution to Problem \ref{LP-problem}.
Applying Claim \ref{claim-feasibility-of-dual}, this implies that
all dual constraint point $\Pi^{*}$, and the origin, are on the same
side of $\left(\pi,\sigma\right)$.
\end{defn}
Since the origin is a strictly feasible point, an optimal solution
$p$ to Problem \ref{LP-problem} must have a positive $z$ value.
As a result, its dual must have a negative $z-$intersect. Moreover,
since $p$ has a largest $z$ value amongst all feasible points, its
dual must have the largest (negative) $z-$intersect amongst all feasible
dual planes. In the case that the dual plane can be made to have an
arbitrarily small negative $z-$intersect, the problem is unbounded.

It follows, then, that a plane which supports the set of points $\Pi^{*}$
from below and has a maximal (negative) $z-$intersect, is a solution
to Problem \ref{LP-problem}, and this is exactly the problem which
the algorithm from the third section of \cite{Grushko2012} solves.

\subsection{Strict Feasiblity of the Origin}

Given a strictly feasible solution $p_{0}$ to Problem \ref{LP-problem},
set $v\triangleq Ap_{0}$, and replace $b$ by $b'=b-v$. Because
$Ap\leq b$ if and only if $A\left(p-p_{0}\right)\leq b'$, the feasible
set of the new problem equals the feasible set of the original problem,
translated by $p_{0}$. In addition, because $v<b$, it holds that
$b'>0$, which means that $A\cdot0<b$. That is, the origin is a strictly
feasible point.

Finding a strictly feasible solution to Problem \ref{LP-problem}
can be performed by solving the following LP problem,
\begin{eqnarray*}
\mbox{minimize}_{s\in\mathbb{R},p\in\mathbb{R}^{d}} &  & s\\
\mbox{s.t.} &  & A\cdot p-b\leq s
\end{eqnarray*}
for which $p=0$ and $s=-\min\left(b\right)+1$ are a feasible solution.
If the optimal solution $s^{*}$ is negative, $p^{*}$ is a strictly
feasible point. 

Alternatively, the equivalent min-max problem can be solved in the
way described in the third section of \cite{Grushko2012},
\[
\min_{p\in\mathbb{R}^{d}}\max\left(Ap-b\right);
\]
if the solution is negative, $p^{*}$ is a strictly feasible point.

\subsection{Rotation}

Let $c\in\mathbb{R}^{d}$ be a general vector, and $u=c-\left(0,0,\dots,1\right)^{T}$.
Define the following matrix:
\[
R'=I-\frac{uu^{T}}{\left\Vert u\right\Vert ^{2}},
\]
and set $R$ to be $R'$ with its first row negated. It is straightforward
to verify that $R$ is a rotation matrix, and that $Rc=\left(0,0,\dots,1\right)^{T}$.

Applying $R$ to a general LP program,

\begin{eqnarray*}
\mbox{maximize}_{x\in\mathbb{R}^{d}} &  & \left(Rc\right)^{T}x\\
\mbox{subject to} &  & \left(AR^{T}\right)x\leq b,
\end{eqnarray*}
results in the form of Problem \ref{LP-problem}. The solutions of
the two problems are related by rotation with $R$.

The computational cost of this rotation is bounded by $\mathcal{O}\left(dn\right)$.

\bibliographystyle{plain}
\bibliography{lp-reduction-new-formulation}

\end{document}